\documentclass[conference]{IEEEtran}
\usepackage[utf8]{inputenc}
\usepackage[ruled,vlined]{algorithm2e}
\usepackage[english]{babel}
\usepackage{csquotes}
\usepackage{amsthm}
\usepackage{amsmath}
\usepackage{amsfonts}
\usepackage{csquotes}
\newtheorem*{remark}{Remark}
\DeclareMathOperator*{\argmax}{arg\,max}

\newtheorem{theorem}{Theorem}
\theoremstyle{definition}
\newtheorem{definition}{Definition}[section]
\newtheorem{lemma}[theorem]{Lemma}
\usepackage{amssymb}
\newcommand{\R}{\mathbb{R}}
\usepackage{graphicx}
\usepackage{mathtools}

\def\BibTeX{{\rm B\kern-.05em{\sc i\kern-.025em b}\kern-.08em
    T\kern-.1667em\lower.7ex\hbox{E}\kern-.125emX}}

\begin{document}

\title{Relay Protocol for Approximate Byzantine Consensus}
\author{
\IEEEauthorblockN{Matthew Ding}
\IEEEauthorblockA{
\textit{Westford Academy, MIT PRIMES}\\
Westford, USA \\
matthewding@berkeley.edu}
}

\maketitle

\begin{abstract}
Approximate byzantine consensus is a fundamental problem of distributed computing. This paper presents a novel algorithm for approximate byzantine consensus, called Relay-ABC. The algorithm allows machines to achieve approximate consensus to arbitrary exactness in the presence of byzantine failures. The algorithm relies on the usage of a relayed messaging system and signed messages with unforgeable signatures that are unique to each node. The use of signatures and relays allows the strict necessary network conditions of traditional approximate byzantine consensus algorithms to be circumvented.

We also provide theoretical guarantees of validity and convergence for Relay-ABC. To do this, we utilize the idea that the iteration of states in the network can be modeled by a sequence of transition matrices. We extend previous methods, which use transition matrices to prove ABC convergence, by having each state vector model not just one iteration, but a set of $D$ iterations, where $D$ is a diameter property of the graph. This allows us to accurately model the delays of messages inherent within the relay system.
\end{abstract}

\section{Introduction}
The idea of byzantine fault-tolerance was first introduced by Lamport et al. \cite{ByzantineGeneral}. Byzantine consensus has since become a large research topic, with applications such as blockchain technology \cite{Blockchain} and machine learning \cite{ByzantineML}. 

Dolev et al. \cite{Dolev} modified and extended the problem of byzantine agreement by introducing \textit{approximate Byzantine agreement}, allowing machines to reach approximate consensus rather than exact consensus. This was motivated by the fact that exact consensus in asynchronous systems was proven to be impossible \cite{Lynch}. Additionally, in synchronous systems, approximate byzantine consensus can be used to create algorithms that do not require complete knowledge of the network topology \cite{IABCPart1}.

This approximate byzantine consensus problem aims to have all honest machines converge to a single state within the convex hull of initial states as the number of iterations approaches infinity \cite{Dolev}. Vaidya \cite{Transition} utilizes a method describing the progression of states in the network using transition matrices to prove consensus.

The main contribution of this paper is to utilize two key tools that have seen a lot of use and success in traditional byzantine consensus problems: messages with unforgeable signatures \cite{DolevStrong, ByzantineGeneral,Dishonest} and relaying messages \cite{ByzantineGeneral,Relay}; our work is a generalization of \cite{Transition}, with signatures and relays used to circumvent certain network assumptions that would otherwise be necessary. To the best of my knowledge, this is the first time either of these two tools have been used in approximate byzantine consensus algorithms.

Byzantine consensus has had applications to machine learning by having machines perform gradient descent steps to minimize a loss function on top of consensus techniques. The combination of delays and signatures may have additional applications in byzantine gradient descent.  \cite{BRIDGE} extends approximate byzantine consensus algorithms for use in byzantine gradient descent methods. We leave the question as to whether Relay-ABC may have a similar extension to our future work.

\section{Problem Formulation}
\subsection{Definitions}
\subsubsection{Graph Notation and Definitions}
\begin{itemize}
    \item Let $m$ be the total number of machines (or "nodes"), $h$ the number of honest machines, and $b$ be the number of byzantine machines ($h+b=m$). 
    
    \item Let $B$ denote the set of byzantine nodes and $H$ denote the set of honest nodes. The set of all nodes $V$ is equal to $B \cup H$.
    
    \item Denote $N_{i}^{I}$ to be the set of all machines that have incoming edges from machine $i$. Denote $N_{i}^{O}$ to be the set of all machines that have outgoing edges to machine $i$
    
    \item Define $dist(i,j)$, for some $i,j \in V$ as the length of the shortest path from node $i$ to $j$
    
    \item In this paper, a "non-zero value" refers to a value that can be lower bounded by some positive constant. In the context of transition matrices, this value doesn't approach 0 as the number of iterations approaches infinity.
\end{itemize}

\subsubsection{Matrix Definitions}
Let $M$ denote some arbitrary matrix, and $M[t]$ denote some arbitrary matrix with respect to iteration $t$.
\begin{itemize}
    \item Transition matrix $M$ denotes a square matrix of size $hD\times hD$
    
    \item $M_{i}[t]$ denotes the $i$th row of matrix $M[t]$
    
    \item $M_{ij}[t]$ denotes the element at row $i$ and column $j$ of matrix $M[t]$
    
    \item We define the first row and column in every matrix as row/column 0
    
    \item Matrix Splicing: $M[a,b:c,d]$ denotes the submatrix spliced by top row $a$, bottom row $b$, left column $c$ and right column $d$ (all inclusive) of matrix $M$
    
    \begin{figure}
    
    $M=\left[\begin{matrix}
    0 & 1 & 2\\
    3 & 4 & 5\\
    6 & 7 & 8\\
    \end{matrix}\right]$
    \centering
    
    \medskip
    $M[0,1:0,1]=\left[\begin{matrix}
    0 & 1\\
    3 & 4\\
    \end{matrix}\right]$
    \centering
    \caption{Example of Matrix Splicing}
    \label{fig:5}
\end{figure}
    
    \item A non-zero column is a matrix column filled entirely with non-zero elements
\end{itemize}

\subsection{Decentralized Learning Model}
We consider a static, directed network $G(V, E)$, where $V = \{0,1,2,... m-1\}$ and $E$ representing communication links between neighboring nodes. If $(i, j) \in E$, then node $i$ may send messages to node $j$. 

Our protocol is analyzed in the synchronous communication setting, where communication occurs over a sequence of iterations. Messages sent during an iteration are guaranteed to be received by the intended recipient within a given finite amount of time. 

Each node $i\in V$ starts with an initial real-valued input. The goal of the protocol is to approach a state that satisfies the following two conditions:
\begin{itemize}
    \item Validity condition: After each iteration of the protocol, the state of each honest node remains within the convex hull of the initial inputs of all honest nodes.
    
    \item Convergence condition: The difference between the states of any two honest nodes approaches zero as the number of iterations approaches infinity
\end{itemize}

\subsection{Byzantine Failure Model}
Among the $m$ machines in the decentralized network, $b$ of them are byzantine machines. Byzantine machines may deviate arbitrarily from the protocol. For example, a byzantine machine may output any arbitrary real value, and send mismatching messages to each of its neighbors. However, a key restriction of byzantine nodes is that they cannot forge signatures of honest users.

\section{Relay-ABC Algorithm}
\subsection{Assumptions}
\theoremstyle{definition}
\begin{definition}{Honest Subgraph:}
Define the honest subgraph as the graph that is formed by removing all byzantine nodes and all edges connected to byzantine nodes in the original graph.
\end{definition}
\begin{itemize}
    \item The number of byzantine machines is strictly less than one-third the total number of machines ($b<\frac{1}{3}m$).
    
    \item We assume the honest subgraph is bidirectionally connected (there exists directed paths from every honest node to every other honest node).
    
    \item We assume the diameter of the honest subgraph is upper-bounded by $D$. 
\end{itemize}

\subsection{Our Contributions}
Our work is an extension of the work done in \cite{Transition}. Our network assumptions are much less restrictive. In particular, we assume no network connectivity assumptions besides the honest subgraph being bidirectionally connected, which is necessary for any algorithm to achieve consensus over all honest nodes. 

The goal of this protocol is to use a relay system to bypass traditional network connectivity assumptions outlined in \cite{IABCPart2}, which has a necessary but insufficient condition that each node has an indegree of at least $2b+1$. This requires that each honest node have at least $b+1$ incoming edges from honest neighbors. On the other hand, bidirectional connectivity of the honest subgraph may possibly be achieved when each honest node has a maximum indegree of as low as one honest neighbor. 

This comes at the tradeoff of higher communication costs, as now machines send to each other at most $m$ sets of parameters in each message, as opposed to one parameter in most other algorithms in literature. 

Our paper provides theoretical guarantees of validity and convergence to an approximate byzantine consensus algorithm with message delays, which to our knowledge has never been done before. Our work also introduces the use of unforgeable signatures. Signatures have seen lots of successful usage in standard byzantine consensus, but until now have not been used in approximate byzantine consensus methods.

\subsection{High-Level Idea}
Since the honest subgraph is bidirectionally connected, this allows all honest machines to receive signed messages from every other honest machine through a broadcast and relay system. Thus we create a pseudo-complete communication graph over the course of an entire phase of $D$ rounds. Since a complete graph does indeed satisfy the necessary conditions of \cite{IABCPart2}, our relay protocol achieves convergence as well.

We use a trimmed-mean aggregation step \cite{BRIDGE,Transition} to ensure byzantine robustness. The trimmed-mean step works by eliminating the greatest $b$ values and the smallest $b$ values, and then taking the arithmetic mean of the remaining values. By removing the greatest and least $b$ values, we ensure that the maximum and minimum values of the set of remaining values are both values of honest machines. This prevents byzantine machines from making the states of honest machines deviate arbitrarily.
\medskip

Each machine $i$ keeps track of their own vector $v_{i}$. This vector consists of $(v_{i}(0), v_{i}(1)... v_{i}(m-1))$, where \textbf{$v_{i}(j)$ represents machine $i$'s most recently updated record of machine $j$'s state.} $v_{i}$ may not always contain a state that was received from an actual message for each machine in the network, as machine $i$ may not have received a message from all machines. 

Machine $i$ only starts performing trimmed-mean steps after $D$ iterations, as this ensures that the first broadcast of all honest machines has had sufficient time to relay across the entire graph and reach every other honest node. This means that each honest machine is outputting an identical message, their initial input value, for the first $D$ iterations.

We choose the specific value $D$, the upper-bound of the diameter of the honest subgraph, so after $D$ iterations, all honest vectors will always contain more honest parameters than byzantine parameters. $D$ is in the worst-case $O(h)$, but with high probability is $O(\log h)$ in Erdos-Renyi random graphs \cite{RandomGraph}.

\medskip

For all honest machines $i$ and $j$, any state $v_{i}(j)$ in the vector will contain a valid signature from machine $j$, as well as an iteration marker, which shows which iteration the parameter was calculated on. Denote $T(v_{i}(j))$ to be the iteration that parameter $v_{i}(j)$ was calculated on.

See Algorithm \ref{Alg1} for the Relay-ABC algorithm.

\begin{algorithm} \label{Alg1}
\SetAlgoLined
    \begin{remark}
        This algorithm is implemented by a specific machine $i$. Each machine $i \in H$ will implement this algorithm concurrently.
    \end{remark}
    
\KwResult{Each state $v_{i}(i)$ converges to the same value within the convex hull of the initial states as Iteration $t \rightarrow \infty$}

    Initialization:
 
    $v_{i}(i)\gets$ Intial State of node $i$ (with signature $i$ and iteration marker $-1$).
    
    \medskip
 
    \For{Iteration $t\gets0$ \KwTo $T$}{
        Broadcast $v_{i}$ to all machines $j\in N_{i}^{O}$
        
        Receive $v_{j}$ from all machines $j\in N_{i}^{I}$
        
        \begin{remark}
            When receiving $v_{j}$, ignore all parameters received that are not properly signed or without a proper iteration marker. If no proper message is from a node, set their incoming value to be an arbitrary predefined real value (e.g. 0).
        \end{remark}

        $G_{i}\gets N_i^O \cup \{i\}$
        
        \For{$j\gets0$ \KwTo $m-1$}{
            \begin{remark}
                In the next two lines, we do the following: Out of all parameters $v(j)$ received from the broadcast step, set $v_{i}(j)$ to the value with the highest iteration marker.
            \end{remark}
            
            \If{$j \neq i$}{
                $g'\gets \argmax_{g:\{g \in G_{i}\}}  T(v_{g}(j))$
                
                $v_{i}(j)\gets v_{g'}(j)$
            }
        }
       
        \If{$t\geq D$}{
            \textbf{Trimmed-mean update step:} 
            
            \medskip
            In a new vector, sort the values of $v_{i}$ in increasing order:
            
            \begin{equation}
                v^*_{i} \gets sort(v_{i})
            \end{equation}
            
            Ignore the least and greatest $b$ values, and set the value of $v_{i}(i)$ to be the average of all remaining values in $v^*_{i}$, as defined below:
            
            \begin{equation} \label{update}
                v_{i}(i) \gets \frac{1}{m-2b} \sum_{k=b}^{m-b-1} v^*_{i}(k)
            \end{equation}
       
            Add signature $i$ and iteration marker $t$ to $v_{i}(i)$
        }
    }
 \caption{Relay-ABC}
\end{algorithm}

\section{Theoretical Guarantees}
\subsection{Overview}
In the following section, we prove that the Relay-ABC algorithm satisfies the validity and convergence conditions.

We define transition matrix $M[t]$ and construct it so that it models the state update of all honest nodes as defined in Algorithm 1. The transition matrices are then used to prove that Algorithm 1 guarantees convergence over all honest nodes.

\subsection{Matrix Definitions}
We introduce several definitions for matrices, most of which are adapted from \cite{Transition}.

To make analysis easier, we will differentiate between \textit{iterations} ($r$) and \textit{phases} ($t$). Define an \textit{iteration} as a single instance in time of communications. During each iteration, every node sends and receives messages from its neighbors. Define a \textit{phase} as a set of $D$ iterations. Therefore the $ith$ phase contains iterations $(i-1)D$ to $iD$. Note that the first iteration of the protocol is iteration 0, while the first phase is phase 1. Since the distance between any two honest nodes is at most $D$, any message from an honest node is guaranteed to reach all other honest nodes within $D$ iterations. We consider phases of $D$ iterations for theoretical analysis, not for any actual implementation in the protocol.

Denote $v[t]$ as the column vector consisting of the states of all honest nodes in phase $t$ (over all $D$ iterations). $\|v[t]\| = hD$, and $v_{rh + i}[t]$ represents the state of node $i$ at iteration $r$ of phase $t$.

Denote $v[0]$ to be the column vector consisting of the initial states of all honest nodes during the first $D$ iterations. Since all machines output the same identical message for the first $D$ iterations, $v[0]$ consists of $D$ identical $h\times1$ column vectors stacked on top of each other.

We express the iterative update of the state of a fault-free node $i \in H$ in any single phase using the matrix form below:
\begin{equation} 
    v_{i}[t] = M_{i}[t-1]*v[t-1]
\end{equation}

The row vector $M_{i}[t]$ satisfies the following conditions:

$M_{i}[t]$ is a stochastic row vector of size $hD$ (proven in Appendix \ref{appendixC}). Thus, $M_{ij}[t]\geq 0$ for $0\leq j \leq hD-1$, and
\begin{equation*}
    \sum_{j=1}^{hD} M_{ij}[t]= 1
\end{equation*}

\medskip

By stacking $hD$ stochastic row matrices $M_{i}[t]$ on top of one another, where M[t] is an $hD \times hD$ matrix, we can represent to state update of all honest nodes in a single matrix multiplication:
\begin{equation} \label{TransitionUpdate}
    v[t] = M[t-1]*v[t-1]
\end{equation}

The matrix M[t] models the state update detailed in Algorithm 1. We detail the specifics of the construction in Section \ref{matrixConstruction}.

Each element in a row of the transition matrix represents some weight of the state column of the previous $D$ iterations. We describe every update of (\ref{update}) of a single node during a single phase as a matrix row.

By repeating the transition matrix update of (\ref{TransitionUpdate}) $T+1$ times, we get:
\begin{equation}
    v[T] = \prod_{t=0}^{T} M[t]*v[0]
\end{equation}

Note that this is an extension of the transition matrix idea of \cite{Transition}. Instead of a $h \times h$ transition matrix and a $h \times 1$ state vector representing the update of a single iteration, we use an expanded transition matrix to describe the update for the set of the next phase ($D$ iterations) using the set of states from the previous phase.

\subsection{Transition Matrix Construction} \label{matrixConstruction}
In this section, we introduce how to construct transition matrices to exactly model the transition of states as dictated by the Relay-ABC algorithm (Algorithm \ref{Alg1}).

We define how to construct a given row of the transition matrix. We introduce some new definitions and notation, most of which is adapted from \cite{Transition}:

\medskip
Let us consider an arbitrary honest node $i$ performing the update step (\ref{update}) at some iteration.
Vector $v_{i}$ is the set of all most updated states from each machine known to machine $i$, and is the set of all values being considered in the trimmed-mean step. It is known that $\| v_{i} \| = m$. Define set $L$ and $S$ to be the largest and smallest $b$ values, respectively, in vector $v_{i}$. Let $N_{i}^{*}$ denote the set all all states that were not removed in the trimmed-mean update step. It is known that $L$ and $S$ are disjoint sets, $|L|=|S|=b$, $N_i^* = v_i - (L\cup S)$, and $\| N_{i}^{*} \| = m-2b$.

Denote $f$ to be the number of faulty states within $N_{i}^{*}$ (faulty states that were not trimmed away.) Define subsets $L^*$ and $S^*$ such that $L^*\subseteq L$, $S^*\subseteq S$, $|L^*|=|S^*|=X$, and that $L^*$ and $S^*$ consist of only honest states (which may be arbitrarily chosen.) It is shown in \cite{Transition} that such subsets always exist.

\begin{algorithm} \label{construction}
\SetAlgoLined
    \begin{remark}
        This algorithm describes how to construct row $t$ of $M$, the $hD \times hD$ transition matrix.
    \end{remark}

    Initialization:
    For all integers $k$ such that $0 \leq k<hD: M_{tk} \gets 0$
    
    \medskip
    
    \uIf{$f=0$}{
        $G \gets $ Case 1 Construction
    }
    \uElse{
         $G \gets $ Case 2 Construction
    }
    
    \medskip
    
    \For{Node $k \in V$}{
        \uIf{$k \neq i$}{
            $iter = t \bmod h - dist(i,k)$
        }
        \uElse{
            $iter = t \bmod h - 1$
        }
        
        $j = iter \bmod D$
            
        $j = j*h + k$
            
        \uIf{$iter < 0$}{
            \begin{equation} \label{sum1}
                M_{tj} = M_{tj} + G_{ik}
            \end{equation}
        }
        \uElse{
             \For{$v \gets0$ \KwTo $hD-1$}{
                \begin{equation} \label{sum2}
                    M_{tv} = M_{tv} + G_{ik}*M_{jv}
                \end{equation}
             }
        }
    }

 \caption{Row Construction Algorithm}
\end{algorithm}

\subsubsection{Matrix Construction Algorithm Overview}
To construct our matrix, we consider two cases: Case 1 where $N_{i}^{*}$ contains no states from faulty nodes ($f=0$), and Case 2 where $N_{i}^{*}$ contains states from at least one faulty node ($f>0$) \cite{Transition}. We describe our transition matrix by using elements of the transition matrices in \cite{Transition}.

Without loss of generality, we denote the first iteration of the given phase corresponding to the transition matrix as iteration 0. This is for simplicity, and in actuality all iterations will be shifted upwards by some positive multiple of $D$. Here we will describe how to construct row $M_t$, which represents iteration $\lfloor\frac{t}{h}\rfloor$, where $t < hD$. Assume without loss of generality that row $t$ is a Node $i$ row. 

In the first iteration of updates of a phase (represented by a single transition matrix), every value $v_{i}(j)$ in (\ref{update}) can be represented exactly by a value in the state vector (a state of the previous phase). However, we note that this is not true for any subsequent iteration: updates in iteration 1 may utilize states of iteration 0, which is a state of the current phase rather than the previous one.

Our paper's main contribution is to extend the transition matrix work of \cite{Transition}: any state from the current phase, rather than the previous state, may be represented as a convex combination of states from the previous phase. Specifically, the state of node $i$ at iteration $r$ of the current phase may be represented as
\begin{equation}
    \sum_{k=0}^{hD-1} M_{(rh + i)(k)} * v_{k}
\end{equation}

We prove that matrix $M$ is stochastic in both Case 1 and Case 2 in Appendix \ref{appendixC}.

\subsubsection{Case 1: $f=0$}
Define matrix $G$ to be the matrix constructed in Section 5.1.1 of \cite{Transition}, using our value of $v_i$ as $N_i^-$. $G$ is an $h\times h$ stochastic matrix.

Algorithm \ref{construction} is motivated by the fact that at iteration $t$, node $i$ is receiving a state from every other node. In particular, node $i$ receives the state of node $j$ from iteration $t - dist(i,j)$, as this will always be the most up-to-date iteration of node $j$ received by node $i$. We construct the matrix such that each state in $N_i^*$ is given an equal weight, as defined in (\ref{update}).

Since each node in $N_i^*$ is honest, each of its states can either be described by an element in a transition matrix, or a row of values in the transition matrix.

\subsubsection{Case 2: $f>0$}
Define matrix $G$ to be the matrix constructed in Section 5.1.2 of \cite{Transition}, using our value of $v_i$ as $N_i^-$. $G$ is an $h\times h$ stochastic matrix.

This construction is motivated by the fact every state in $N_{i}^*$ can we represented by a weighted average of two honest nodes, one in $L^*$ and one in $S^*$. This allows even the behavior of byzantine nodes to be able to be represented in the transition matrix of honest states.

\subsection{Validity Proof}
The update in (\ref{update}) of each node always results in a convex combination of some set of node states during some iterations. This means that any update will always stay within the convex hull of the set of initial input states, proving the validity condition.

\subsection{Convergence Proof}
\subsubsection{Matrix Characteristics}
\theoremstyle{definition}
\begin{definition}{Node $i$ Row:}
row $j$ of Matrix $M$ is considered a Node $i$ Row iff $j \bmod h \cong i$
\end{definition}

A Node $i$ Row represents the state update of node $i$ during some iteration of the phase.

To characterize what the weights of matrix rows representing iterations after iteration 0, we introduce the following observation:
\begin{theorem} \label{T1}
    If an element of the transition matrix $M_{ij}$ is a non-zero weight for $i < h$, then the value of $M_{zj}$ is non-zero as well, $\forall z$ such that $z$ is a Node $i$ row and $z\geq h$
\end{theorem}

\begin{proof}
    Row $i$ represents the matrix update of node $i$ in iteration 0, while row $z$ represents the matrix update of node $i$ in any iteration $[1,D]$. We denote $z = kh + i$, for some integer $k<D$.
    
    We proof the theorem with induction:
    
    Base case ($k=0$): If $k=0$, then $z=i$. The theorem is trivially true.
    
    Induction Step ($0<k<D$): We noted previously that not every update can be exactly expressed as a convex combination of weights from the previous phase. Specifically, node $i$ will always use the most up-to-date value of its own state, which is no longer a state of the previous phase. 
    
    However, this state may still be represented as a \textit{combination} of weights of the previous phase: in iteration $t$, node $i$ uses state of node $i$ in iteration $t-1$. Iteration $t-1$ of phase $a$ is represented by row $h(t-1)+i$ of matrix M:
    \begin{equation}
        v_{i}(i) = M_{h(t-1) + i} * v[a]
    \end{equation}
    
    Therefore, instead of using a single weight to denote the state of node $i$ in iteration $t-1$, we may instead add every single weight of row $h(t-1)+i$ to row $ht + i$, scaled by a factor of $\frac{1}{hD}$. The induction step is completed by setting $k$ as t.
\end{proof}

We now describe specific characteristics of which weights are non-zero in transition matrix $M$:

We first note the existence of a diagonal of non-zero values at the rightmost $h$ columns of the transition matrix.
\begin{theorem} \label{T2}
   $\forall k, i$ such that $k, i \in \R$, $0\leq k < D$ and $0 \leq i \leq h$, $M_{kh + i, h(D-1) + i}$ is a non-zero value.
\end{theorem}

\begin{proof}
    For $k=0$, row $kh+i$ represents the update of node $i$ during iteration 0, or the first iteration of the previous phase. Column $h(D-1) + i$ represents the state of node $i$ during iteration $D-1$ of the previous phase, or the last iteration of the previous phase. Since each honest node always uses the state of the previous iteration in its update, this matrix value is non-zero.
    
    For $k>0$, the result generalizes from Theorem \ref{T1}.
\end{proof}

\begin{figure} [ht]

$\left[\begin{matrix}
0 & 0 & \frac{1}{3} & \boldsymbol{\frac{1}{3}} & \frac{1}{3} & 0\\[6pt]
0 & 0 & 0 & \frac{1}{3} & \boldsymbol{\frac{1}{3}} & \frac{1}{3}\\[6pt]
\frac{1}{3} & 0 & 0 & 0 & \frac{1}{3} & \boldsymbol{\frac{1}{3}}\\[6pt]
0 & 0 & \frac{1}{9} & \boldsymbol{\frac{2}{9}} & \frac{2}{9} & \frac{4}{9}\\[6pt]
\frac{1}{9} & 0 & \frac{1}{9} & \frac{2}{9} & \boldsymbol{\frac{1}{3}} & \frac{2}{9}\\[6pt]
\frac{1}{9} & 0 & 0 & \frac{4}{9} & \frac{2}{9} & \boldsymbol{\frac{2}{9}}
\end{matrix}\right]$
    \centering
    \caption{A sample matrix illustrating Theorem \ref{T2}, with $h=3$ and $D=2$ is shown.}
    \label{fig:1}
\end{figure}

Now we introduce some definitions relating to the network graph.

\theoremstyle{definition}
\begin{definition}{Complete Graph:}
A complete graph is a graph with vertex set $V$, and edge set $E'$, such that $\forall i,j, i \neq j: (i,j) \in E'$. The graph contains $b$ byzantine nodes, $h$ honest nodes, and $\|V\| = m$.
\end{definition}

A complete graph describes the de-facto communication during an entire phase: since the longest path between any two honest nodes is at most $D$, any two nodes may communicate with each other for at least one iteration during every single phase. We now introduce a graph that represents the network graph after the trimming in (\ref{update}).

\theoremstyle{definition}
\begin{definition}{Reduced Graph:}
A reduced graph is a complete graph with all nodes in set $B$ removed, along with their incoming and outgoing edges. Additional, we remove any arbitrary set of $b$ incoming edges from each remaining node. Note that there are several reduced graph for every complete graph, but only a finite number of them. Define $R_{f}$ to be the set of all reduced graphs for a given complete graph, and define $r$ as $\| R_{f} \|$. Note that this definition comes from \cite{Transition}. 

Note that even though nodes do not have edges connecting to themselves, they can also "send" messages to themselves. Thus the adjacency matrix $A$ of a reduced graph will always be non-zero at $A_{ii}, \forall i\in V$
\end{definition}

The reduced graph represents the communication links between the entire graph after trimming is done.

We introduce one last theorem describing the qualities of transition matrix $M$. 
\begin{theorem} \label{T3}
   Every Node $i$ row contains a non-zero value in column $z$, where $z \bmod h$ is a node with an incoming edge of node $i$ in some reduced graph in $R_{f}$
\end{theorem}
\begin{proof}
    For rows $t$ such that $\lfloor\frac{t}{h}\rfloor = 0$ (which correspond to the updates of node $i$ in iteration 0), we note that every single node received some message from every other node from the previous phase. The reduced graph represents some form of trimming of all incoming information, where an incoming edge of the reduced graph represents an incoming state that is not trimmed. Column $z \bmod h \cong j$, column $z$ represents a state of node $j$. The theorem is proven by induction. 
    
    \paragraph{Base Step}
    The base step is proven in \ref{appendix A}
    
    \paragraph{Induction Step}
    For rows $t$ such that $\lfloor\frac{t}{h}\rfloor > 0$, the result generalizes from Theorem \ref{T1}
\end{proof}

This theorem describes how every state update in the matrix is based off of at least one weight corresponding to each node in the graph.

\subsubsection{Transition Matrix Behavior}
To prove convergence, we show that a finite number of matrices in the transition matrix product forms a non-zero column in the stochastic matrix (scrambling matrix). This guarantees that the matrix $\lim_{T\to\infty} \prod_{t=0}^{T} M[t]$ has identical rows, which in turn ensures the convergence condition \cite{Transition, Wolfowitz}.

To do this, we introduce a repeated matrix product $Q[t]$, which represents a repeated product of $2rD+1$ matrices. Specifically, we define the following:
\begin{equation}
    Q[t] = \prod_{i=t(2rD+1)}^{(t+1)(2rD+1)} M[i]
\end{equation}

We may write our transition matrix update thus as:
\begin{equation}
    v[T] = \prod_{t=0}^{T / (2rD+1)} Q[t]*v[0]
\end{equation}

Now we introduce a new theorem that explains the behavior of transition matrices. 
 
 \theoremstyle{definition}
\begin{definition}{Matrix Inequality:}
We define matrices $A < B$ iff $\forall i, j: A_{ij} < B_{ij}$.
\end{definition}

\begin{theorem} \label{T4}
   $\forall t$, matrix $Q[t]$ contains at least one non-zero column within the last $n$ columns of the matrix.
\end{theorem}
\begin{proof}
   We introduce Lemma \ref{L1}, which explains how the product of two transition matrices creates a $h \times h$ adjacency matrix of a reduced graph at the bottom right corner of the matrix.
 
 \begin{lemma} \label{L1}
    Define $M^1 * M^2$ to be two arbitrary transition matrices constructed from Algorithm \ref{construction}. $M^1 * M^2[h(D-1) + 1, hD : h(D-1) + 1, hD] \geq \beta*r_{f}$, where $r_{f}$ is the adjacency matrix of some reduced graph and $\beta$ is some positive constant \cite{Transition}.
 \end{lemma}
 \begin{proof}
    The proof is derived from Theorem \ref{T3}. Within any row of matrix $M^1$, if there is a non-zero value in column $z$, where $z \bmod h \cong j$, then there will be a non-zero value in column $h(D-1) + j$ of matrix $M^1 * M^2$. This is due to diagonals of Theorem \ref{T2} in matrix $M^2$.
 \end{proof}
 
 We now use a key result of \cite{Transition} to show the generation of a partial non-zero column in this bottom-right matrix.
 
 \begin{lemma} \label{L2}
    The product of $2rD$ arbitrary transition matrices (denoted as $Z$) results in a non-zero column of the matrix $Z[h(D-1) + 1, hD : h(D-1) + 1, hD]$.
 \end{lemma}
 \begin{proof}
     From Lemma \ref{L1}, the product of $2rD$ transition matrices can be represented as a product of $rD$ matrices, where each matrix has some arbitrary adjacency matrix of a reduced graph in its bottom rightmost $h \times h$ submatrix. Given that $\|R_{f}\| = r$, by the pigeonhole principle, we can conclude that at least one arbitrary reduced graph is repeated at least $D$ times. In Appendix \ref{appendixB}, it is also proven that there exists a directed path from some node to all other nodes in all reduced graphs. Since the shortest directed path of between any two honest nodes is at most length $D$, it is proven in \cite{Transition} that a non-zero column is formed in matrix $Z[h(D-1) + 1, hD : h(D-1) + 1, hD]$
 \end{proof}
 
 Given that matrix $Z$ has a non-zero column in the bottom rightmost $h \times h$ matrix, it can be shown that $Z*M$, where $M$ is an arbitrary transition matrix, creates a non-zero column of the entire matrix in the same column where the $h \times h$ matrix column was. This is once again due to the diagonals of Theorem \ref{T2} and a simple application of linear algebra.
 
 Given that $Z$ represents the product of $2rD$ matrices and $M$ represents a single transition matrix, we have shown that the product of $2rD+1$ transition matrices creates a non-zero column. This concludes the proof of Theorem \ref{T4}.
\end{proof}

\begin{figure}

$\left[\begin{matrix}
0 & \frac{1}{3} & \frac{1}{3} & \frac{1}{3} & 0 & 0\\[6pt]
0 & 0 & 0 & \frac{1}{3} & \frac{1}{3} & \frac{1}{3}\\[6pt]
\frac{1}{3} & 0 & \frac{1}{3} & 0 & 0 & \frac{1}{3}\\[6pt]
0 & 0 & \frac{1}{9} & \frac{2}{9} & \boldsymbol{\frac{2}{9}} & \frac{4}{9}\\[6pt]
\frac{1}{9} & 0 & \frac{1}{9} & \frac{2}{9} & \boldsymbol{\frac{1}{3}} & \frac{2}{9}\\[6pt]
\frac{1}{9} & 0 & 0 & \frac{4}{9} & \boldsymbol{\frac{2}{9}} & \frac{2}{9}
\end{matrix}\right]$
    \centering
    \caption{A sample matrix $Z$ with a non-zero column in the bottom rightmost $h \times h$ matrix.}
    \label{fig:2}
\end{figure}

\begin{figure}
$\left[\begin{matrix}
0 & 0 & \frac{1}{3} & \frac{1}{3} & \boldsymbol{\frac{1}{3}} & 0\\[6pt]
0 & 0 & 0 & \frac{1}{3} & \boldsymbol{\frac{1}{3}} & \frac{1}{3}\\[6pt]
\frac{1}{3} & 0 & 0 & 0 & \boldsymbol{\frac{1}{3}} & \frac{1}{3}\\[6pt]
0 & 0 & \frac{1}{9} & \frac{2}{9} & \boldsymbol{\frac{2}{9}} & \frac{4}{9}\\[6pt]
\frac{1}{9} & 0 & \frac{1}{9} & \frac{2}{9} & \boldsymbol{\frac{1}{3}} & \frac{2}{9}\\[6pt]
\frac{1}{9} & 0 & 0 & \frac{4}{9} & \boldsymbol{\frac{2}{9}} & \frac{2}{9}
\end{matrix}\right]$
    \centering
    \caption{A sample matrix $Z*M$ with a non-zero column in the entire matrix in the same column as the partial non-zero column of Figure \ref{fig:2} above.}
    \label{fig:3}
\end{figure}

\begin{theorem} \label{T7}
   $\lim_{T\to\infty} v[T] = c* \textbf{1}$, where $c$ is some constant and $\textbf{1}$ is the column vector of ones. Note this proves the convergence condition.
\end{theorem}
\begin{proof}
    \begin{align*}
       \lim_{T\to\infty} v[T] = \\
       \lim_{T\to\infty} \prod_{t=0}^{T} M[t]*v[0] = \\
       \lim_{T\to\infty} \prod_{t=0}^{T / (2rD+1)} Q[t]*v[0]
    \end{align*}

    From Theorem \ref{T4}, we have shown that $\forall t$, matrix $Q[t]$ contains at least one non-zero column. Since $Q[t]$ is also stochastic, it is a scrambling matrix. \cite{Transition} proves that the product of any infinite number of scrambling matrices converges to a matrix with identical rows. Thus the product $\prod_{t=0}^{T / (2rD+1)} Q[t]*v[0]$ results in a column vector with identical elements, proving the theorem.
\end{proof}

\section*{Acknowledgements}
First of all, I would like to thank MIT and the MIT PRIMES program for giving me this wonderful opportunity to conduct research.

I'd like to thank Jun Wan, Prof. Lili Su, and Prof. Nitin Vaidya for all of our discussions. Finally, the biggest thanks goes out to my research mentor, Hanshen Xiao, for all of his support and guidance.

\bibliographystyle{IEEEtran}
\bibliography{main.bib}

\appendix
\section{Reduced Graph Proof} \label{appendix A}
For each row $i$ of an arbitrary transition matrix $M$, we seek to prove that at least $h-b+1$ elements in $M_i$ are lower bounded by some arbitrary positive constant $\beta$ \cite{Transition}. We do this through a proof by induction.

\subsection{Base Step: $\lfloor\frac{t}{h}\rfloor = 0$}
In Algorithm \ref{construction}, a value of $i$ such that $\lfloor\frac{i}{h}\rfloor = 0$ implies that row $i$ models a node's update in iteration 1 of the phase. This also implies that $\forall k, iter<0$.

$\forall k$, column $j$ is a unique value. This implies that through each iteration of the for-loop over all nodes, a unique column is being considered. We now consider two additional cases (whether matrix $G$ is Case 1 or Case 2 construction), and prove that they both satisfy the desired condition.

\subsubsection{Case 1}
$G_{ik}>\beta$ iff $k\in N_i^*\cup i$ \cite{Transition}. Since $N_i^* \subset V$, and $\|N_i^*\cup i\| = m-2b+1 = h-b+1$, we can conclude that at least $h-b+1$ elements of $M_i$ are lower-bounded by $\beta$.

\subsubsection{Case 2}
$G_{ik}>\beta$ iff $k\in (N_i^*\cap H)\cup i\cup L^* \cup S*$ \cite{Transition}. In \cite{Transition}, it is shown that $(N_i^*\cap H)\cup i\cup L^* \cup S* = \|m \cap H\| -b - 1$. Since $(N_i^*\cap H)\cup i\cup L^* \cup S* \subset V$, and $\|m \cap H\|-b-1 = h-b+1$, we have concluded the prove of the base case.

\subsection{Reduced Graph Inequality}
Each node in a reduced graph has a total of $m-2b$ incoming edges. When you include a nodes ability to communicate with itself, each row of the adjacency matrix of the reduced graph contains $m-2b+1 = h-b+1$ ones. We now note that from Algorithm \ref{construction}, $M_{tj}$ is non-zero only if $(t \bmod h, j \bmod h) \in E$, and that at most one non-zero value exists on some column $j$ for each value of $j \bmod h$. This concludes the proof of the base step of the theorem.

\section{Source Component Proof} \label{appendixB}
We define a source component of a graph as a node that has a directed path to every other path in a graph. In this section we prove that any arbitrary reduced graph contains at least one source component.

A reduced graph is constructed by removing $n$ incoming edges from each node of a fully connected directed graph of at least $2n+1$ nodes. In this proof, we assume that the reduced graph has exactly $2n+1$ nodes. The case where the number of nodes exceeds $2n+1$ is a simple generalization.

In a fully connected graph of $2n+1$ nodes, there are totally (2n+1)2n outgoing edges. Thus, after removing $n$ incoming edges (which are also outgoing edges of some other arbitrary node) from each node, there still exists at least $(2n+1)n$ outgoing edges left in the graph. By Pigeonhole Principle, at least one node in the reduced graph, let us denote it as $v_0$, has at least $n$ outgoing edges. 

Denote the set of all nodes with direct incoming edges from $v_0$ as set $S$. We know that $\|S\|\geq n$. For each of the nodes which are not in the set $S$ (there are no more than n nodes not in S), it is noted that each of them has $n$ incoming edges. 

Assume that none of these nodes have incoming edges from $v_0$ or any node in set $S$. Thus, they can only have edges from at most a total of $2n+1 - 2 - n = n-1$ nodes. However, it is known that all nodes have $n$ incoming edges. This is a contradiction. Thus, they must either have one incoming edge from a node in $S$, or an incoming edge from $v_0$. Therefore we have proven that $v_0$ is a source component.

\section{Transition Matrix $M$ is stochastic} \label{appendixC}
We prove that $M$ is stochastic through strong induction. Let $t$ be an arbitrary row of matrix $M$. We prove that the sum of all elements in $M_t$ is 1.

\subsection{Base Case: $\lfloor\frac{t}{h}\rfloor = 0$}
In Algorithm \ref{construction}, a value of $t$ such that $\lfloor\frac{t}{h}\rfloor = 0$ implies that row $t$ models a node's update in iteration 0 of the phase. This also implies that $\forall k, iter<0$. 

Throughout the entire for-loop over Node $k \in V$, only summation step \ref{sum1} is used. Each summation step adds $G_{ik}$ to the value of $\sum_{k=0}^{hD-1} M_{tk}$, for some $k$. Thus the entire for-loop adds a total of 
\begin{equation}
    \sum_{k=0}^{m-1} G_{ik}
\end{equation}
to $\sum_{k=0}^{m-1} G_{ik}$, split over multiple elements. Since $\sum_{k=0}^{m-1} G_{ik}=1$, $M_t$ is a stochastic vector.

\subsection{Base Case: $\lfloor\frac{t}{h}\rfloor > 0$}
Without loss of generality, assume that $\lfloor\frac{t}{h}\rfloor=\alpha$, for some $\alpha>0$. By strong induction, we assume that $M_g$ is a stochastic vector, for all $\lfloor\frac{g}{h}\rfloor<\alpha$.

Summation step (\ref{sum1}) adds $G_{ik}$ to the value of $\sum_{k=0}^{hD-1} M_{tk}$, for some $k$. However, since $\lfloor\frac{j}{h}\rfloor<\lfloor\frac{t}{h}\rfloor=\alpha$, we know that $M_{jv}$ is a stochastic vector. Therefore summation step (\ref{sum2}) adds $G_{ik}$ to the value of $\sum_{k=0}^{hD-1} M_{tk}$, for some $k$. Therefore once again, the entire for-loop adds a total of 
\begin{equation}
    \sum_{k=0}^{m-1} G_{ik}
\end{equation}
to $\sum_{k=0}^{m-1} G_{ik}$, split over multiple elements. Similar to the base case, this proves that $M_t$ is a stochastic vector, completing the proof.

\end{document}